%% file: arxiv_main.tex
\documentclass[twoside,10pt]{article}

\usepackage{xcolor}
\usepackage{hyperref}
\hypersetup{colorlinks=true,linkcolor=violet,citecolor=violet,urlcolor=violet}
\usepackage{jmlr2e}
\usepackage{graphicx}
\usepackage{amsmath}
\usepackage{cleveref}
\usepackage{mathtools}
\usepackage{xspace}
\usepackage{multicol}
\usepackage{multirow}
\usepackage{booktabs}
\usepackage{caption}
\usepackage{subcaption}
\usepackage{algorithm}
\usepackage{algorithmic}

\usepackage{arxiv_notation}
\newcommand{\model}{{FairExp}}

\def \xijt {{\mathbf{x}_{ij}^t}}

\def \xmns {{\mathbf{x}_{mn}^s}}

\def \ymns {{\mathbf{y}_{mn}^s}}

\def \btheta {\boldsymbol{\theta}}

\def \bigO {\mathcal{O}}

\jmlrheading{1}{2000}{1-48}{4/00}{10/00}{Yiling Jia, Hongning Wang}

\ShortHeadings{Calibrating Explore-Exploit Trade-off for Fair Online Learning to Rank}{Jia and Wang}
\firstpageno{1}

\begin{document}

\title{Calibrating Explore-Exploit Trade-off for Fair Online Learning to Rank}

\author{\name Yiling Jia \email yj9xs@virginia.edu \\
	\addr Department of Computer Science\\
	University of Virginia\\
	Charlottesville, VA 22903, USA
	\AND
	\name Hongning Wang \email hw5x@virginia.edu\\
	\addr Department of Computer Science\\
	University of Virginia\\
	Charlottesville, VA 22903, USA
}

\maketitle

\begin{abstract}

Online learning to rank (OL2R) has attracted great research interests in recent years, thanks to its advantages in avoiding expensive relevance labeling as required in offline supervised ranking model learning. 
Such a solution explores the unknowns (e.g., intentionally present selected results on top positions) to improve its relevance estimation. This however triggers concerns on its ranking fairness: different groups of items might receive differential treatments during the course of OL2R. 
But existing fair ranking solutions usually require the knowledge of result relevance or a performing ranker beforehand, which contradicts with the setting of OL2R and thus cannot be directly applied to guarantee fairness.

In this work, we propose a general framework to achieve fairness defined by group exposure in OL2R. 
The key idea is to calibrate exploration and exploitation for fairness control, relevance learning and online ranking quality.
In particular, when the model is exploring a set of results for relevance feedback, we confine the exploration within a subset of random permutations, where fairness across groups is maintained while the feedback is still unbiased. Theoretically we prove such a strategy introduces minimum distortion in OL2R's regret to obtain fairness.
Extensive empirical analysis is performed on two public learning to rank benchmark datasets to demonstrate the effectiveness of the proposed solution compared to existing fair OL2R solutions.
\end{abstract}

%% This command processes the author and affiliation and title
%% information and builds the first part of the formatted document.
\maketitle

\input{arxiv_intro}

\input{arxiv_related}

\input{arxiv_problem}
\input{arxiv_method}

\input{arxiv_experiment}

\section{Conclusion}
Existing OL2R solutions focus on optimizing user-oriented utilities of the ranked list, but ignore the potential unfair treatment that result ranking can introduce on the content providers. Restricted by the prerequisite on the availability of relevance labels, existing fair ranking solutions for offline  models cannot be applied to the online setting. In this work, we propose a general framework FairExp for fair OL2R. During the interaction with the users, the trade-off between exploration and exploitation in OL2R is calibrated with respect to the fairness constraint. By taking advantage of the structure of deterministic exploration in PairRank, fairness control, relevance learning and online quality can be best optimized simultaneously. We proved that our strategy intrudces the minimum distortion in OL2R's regret to obtain fairness. And our empirical evaluations demonstrate the strong advantage of FairExp in balancing exploitation, exploration and fairness in OL2R.

Currently we focus on OL2R solutions that are based on deterministic exploration strategies, e.g., PairRank is based on construction of confidence intervals, it is interesting to extend \model{} to enforce fairness constraints on those stochastic strategies, e.g., sampling based OL2R solutions. Moreover, the application of \model{} is beyond fair OL2R; and it can be generally applied to handle multi-objective optimization in result rankings, e.g., relevance and fairness can be considered as two distinct objectives for OL2R. Other objectives, such as diversity and trustworthiness, can also be enforced using \model{}. 

\begin{acks}
This paper is based upon the work supported by the National Science Foundation under grant IIS-1553568 and IIS-2128019, and Google Faculty Research Award.
\end{acks}

\bibliography{sample-base}
\end{document}

%% file: arxiv_intro.tex
\section{Introduction}
Result ranking is at the core of many information retrieval applications, such as document retrieval, online advertisement, and recommendations. To quickly capture users' information need and avoid expensive relevance labeling as required in offline solutions, online learning to rank (OL2R) \citep{yue2009interactively,wang2019variance,jia2021pairrank} has recently attracted great research interests. OL2R directly learns from implicit user feedback, such as clicks. It empowers modern retrieval systems to optimize their performance directly from users' interactions and makes the supervised learning possible when collecting explicit annotations from experts is economically infeasible.

Existing OL2R solutions target at user-focused utility optimization, which unfortunately ignores the impact of result ranking on the content providers, who will receive differential attention from the users depending on their results' ranked positions. More specifically, it is widely known that the position of an item in the ranked list has a crucial influence on its exposure and chance to be consumed by the users, e.g., position bias \citep{joachims2005accurately,agichtein2006improving,guan2007eye}. Even a minor difference in relevance can translate into huge discrepancy in exposure across groups and thus societal or economic impact \citep{biega2018equity, singh2018fairness}. For example, job postings ranked on top positions of LinkedIn's search result pages are more likely to be examined and considered by most applicants; as a result, those employers can gain an edge in the business competitions \citep{geyik2019fairness}. 

Fairness does not only pertain to OL2R; instead, any ranking solutions (online or offline) should concern this important aspect. But OL2R brings in unique new challenges for fair result ranking: comparing to the scenario where a pre-trained ranker is deployed, the differential treatment can be accumulated and amplified in a faster rate during the course of OL2R, since the algorithm has to continuously place an intentionally selected subset of results on top. This is driven by the need of exploration in OL2R.
Though various exploration strategies have been proposed in OL2R literature, including those exploring in model space \citep{yue2009interactively, hofmann2012estimating, zhao2016constructing, wang2018efficient, wang2019variance} and action space \citep{katariya2016dcm,zoghi2017online, lattimore2018toprank, li2018online, kveton2018bubblerank,oosterhuis2018differentiable, jia2021pairrank}, none of them consider fairness when presenting results to users. When the exploration is deterministic, e.g., confidence interval based methods \citep{li2018online,jia2021pairrank,lattimore2018toprank}, the situation can become even worse, especially when the confidence estimation is correlated with protected attributes of results. 
This introduces a new conflict in the already complicated explore-exploit dilemma in OL2R.
For instance, to ensure fairness, some results cannot be displayed, which directly slows down or even prevents online model learning. The slow improving relevance estimation in turn can lead to poor user experience (i.e., higher regret); in the meanwhile, if the bias from earlier OL2R update cannot be quickly eliminated by new feedback, unfair result ranking will be accumulated and amplified. This leads to a new paradox among three elements: fairness, exploration and exploitation. 

Most of the existing fair ranking solutions unfortunately do not apply to OL2R, because they require expert-labeled relevance data or logged feedback for all items to train a ranker beforehand. 
For example, in-processing methods impose fairness at training time as an additional optimization objective \citep{zehlike2020reducing, singh2019policy, yadav2021policy}; while post-processing approaches re-rank the result subject to predefined fairness constraints \citep{celis2017ranking, geyik2019fairness, zehlike2017fa}. Their prerequisite on the availability of relevance labels restricts these solutions in the OL2R setting, where the ranker is trained via the interactions with the users on the fly. Recently, a merit-based fair ranking solution is proposed for dynamic learning to rank~\citep{morik2020controlling}, where inverse propensity scoring is applied to estimate the relevance score and a proportional controller is used to mitigate unfairness in result ranking on the fly. But its fairness control part does not consider the uncertainty of the estimated relevance, which can still lead to unfair ranking before the online estimation converges. 

We propose a general framework to achieve fairness defined on group exposure during OL2R. 
The key idea is to calibrate the trade-off between exploration and exploitation under fairness constraints. In particular, when the model is exploring a set of results for relevance feedback, the exploration is confined within a subset of random permutations, where fairness across groups is maintained while the feedback is still unbiased. When the model is exploiting, fairness is directly enforced. 
More specifically, in each round of result serving, a set of group-level placement templates for the top-$k$ positions will first be constructed based on the currently accumulated unfairness. Based on the expected exposure of each rank position under the current query, all templates need to maintain unfairness under a threshold. As fairness on group exposure only concerns the placement of each group (i.e., their ranking order), the construction of such templates is independent from ranker's relevance estimation. Then, the ranked list will be generated with respect to the chosen template and the required exploration by the OL2R algorithm. When multiple templates are qualified, the one with the projected minimum OL2R regret will be chosen (i.e., exploitation). 
In other words, fairness is treated as a hard constraint when striking the balance between exploration and exploitation.

In this paper, we cast a state-of-the-art OL2R algorithm, PairRank \cite{jia2021pairrank}, into a fair online ranking solution. We choose PairRank is only because of its promising empirical performance and feasibility for theoretical analysis. Our proposed framework is general and can be applied to OL2R solutions with action space exploration, as long as their exploration is deterministic (i.e., where to explore can be explicitly computed). And a typical family of such solutions are the confidence interval based methods, such as TopRank \cite{lattimore2018toprank}, RecurRank \cite{li2018online} and PairRank. 
Extensive empirical analysis of the resulting algorithm is performed on two public learning to rank datasets to demonstrate the effectiveness and advantages of our proposed framework comparing to existing fair OL2R solutions.

%% file: arxiv_related.tex
\section{Related Work}
\label{sec_related}
\noindent\textbf{Online learning to rank.}
Based on the different exploration strategies developed in OL2R literature, existing OL2R solutions can be broadly categorized into two main classes: action space exploration and model space exploration.
The first type deliberately controls result ranking at each round to collect unbiased relevance feedback for model update. Due to the combinatorial nature of ranking, different decompositions are proposed to reduce the exponentially large action space.
One category of such solutions find the best ranked list for each individual query separately, by modeling users' click and examination behaviors with multi-armed bandit algorithms \citep{radlinski2008learning,kveton2015cascading,zoghi2017online,lattimore2018toprank}. The exploration is then reduced to every rank position under each query. To generalize the learning across queries, another category of such methods employ parametric models for relevance estimation. For example, PDGD \citep{oosterhuis2018differentiable} samples the next ranked document from a Plackett-Luce model and estimates gradients from the inferred pairwise result preferences. PairRank \citep{jia2021pairrank} explores in the pairwise document ranking space for learning a pairwise logistic regression ranker online.

The second type of OL2R solutions explore in the entire model space to locate the best ranker \citep{yue2009interactively,li2018online,oosterhuis2018differentiable}. The most representative work is Dueling Bandit Gradient Descent (DBGD) \citep{yue2009interactively,schuth2014multileaved}, which uniformly explores in the entire model space based on interleaved feedback. Subsequent methods improved DBGD by developing more efficient sampling strategies, such as multiple interleaving and projected gradient, to reduce variance \citep{hofmann2012estimating,zhao2016constructing,oosterhuis2017balancing, wang2018efficient, wang2019variance}. However, as fairness is not considered in any of those aforementioned OL2R solutions at all, there is no guarantee for fair result ranking, even when some randomization is employed. In this work, we empirically verified this in our experiments.

\noindent\textbf{Fair ranking.}
Fairness in ranking is concerned with an insufficient presence or a consistently differential treatment over different groups across all ranking positions~\citep{castillo2019fairness}. Various definitions of fairness in ranking have been proposed, together with their fair ranking solutions~\citep{zehlike2021fairness}. Most metrics are motivated by those defined in fair classification problems~\citep{mehrabi2021survey}. Specifically, solutions focusing on group parity enforces a proportional allocation of exposure between groups. For example, \citet{yang2017measuring} proposed to reduce the difference in occurrences of different groups on a subset of ranking positions. And a series of work proposed to set a limit on the number of items from each group in the top-$k$ positions \citep{celis2017ranking, geyik2019fairness, zehlike2017fa}. On the other hand, merit-based fairness of exposure allocates exposure to groups based on their merit instead of the group size. Both in-processing \citep{zehlike2020reducing, singh2019policy, yadav2021policy} and post-processing approaches \citep{biega2018equity, singh2018fairness} are proposed to achieve this type of fairness in ranking. However, all the methods assume either the expert-labeled relevance or logged implicit feedback are available for model training or result ranking beforehand. This unfortunately prevents the application of such fair ranking solutions in the online setting.

The most relevant work to ours is the FairCo algorithm proposed in \citep{morik2020controlling}, which applies proportional control to mitigate unfairness in dynamic learning to rank.
Its online ranker update is achieved by the inverse propensity scoring method on user clicks. This however imposes a strong assumption on the user examination behavior (i.e., position-based examination). And when enforcing fairness, uncertainty in the ranker's relevance estimation is not considered. As a result, inaccurate fairness measure might lead to unfair control in this solution. Additionally, the provided theoretical analysis is limited to a fixed set of items, which cannot be applied to the unseen queries and items.
Some recent work studied fairness under inference uncertainty but in somehow different settings. \citet{ghosh2021fair} studied the impact from the uncertainty in sensitive attributes on fair ranking. And \citet{singh2021fairness} focused on the uncertainty in the estimated merit and individual fairness in the ranking problem. In this work, we study group fairness defined by the exposure on the item side in OL2R, where we assume the grouping of items is given and new queries and items can emerge at any time.

%% file: arxiv_problem.tex
\section{Problem Formulation}
In this section, we provide a brief overview about the general problem setting in OL2R, especially about its key challenges, and the notion of fairness in ranking employed in this paper.

\subsection{Online Learning to Rank}

In OL2R, at round $t \in [T]$, the ranker receives a query $q_t$ and its associated $L_t$ candidate documents\footnote{We will use ``document'' and ``item'' interchangeably in this paper to denote the instances to be ranked under a query.} represented by a set of $d$-dimensional query-document feature vectors: $\mathcal{X}_t = \{\xb_1^t, ..., \xb_{L_t}^t\}$ where $\xb^t_i \in \RR^d$. 
The ranking $\pi_t = \big(\pi_t(1), ..., \pi_t(L_t)\big) \in \Pi([L_t])$ is generated by the ranker based on its knowledge so far and its exploration strategy, where $\Pi([L_t])$ represents the set of all permutations of $L_t$ documents and $\pi_t(i)$ is the rank position of document $i$. The user examines the returned ranked list and provides his/her feedback, i.e., clicks $C_t = \{c_1^t, c_2^t, ..., c_{L_t}^t\}$, where $c_i^t = 1$ if the user clicked on document $i$ at round $t$; otherwise $c_i^t = 0$. According to the collected implicit feedback, the ranker updates itself for improved ranking in the next round.

As the ranker learns from user feedback while serving, cumulative regret is an important metric for evaluating an OL2R solution. Usually regret is defined as the cumulative expected difference between the utility of the presented ranked list $\pi_t$ and that of the ideal one $\pi_t^*$ over a finite time horizon $T$, i.e., 
\begin{equation}
\label{eq_regret_def}
    R_T = \mathbb{E}\left[\sum\nolimits_{t=1}^T r_t\right] = \mathbb{E} \left[\sum\nolimits_{t=1}^T f(\pi_t^*) - f(\pi_t)\right],
\end{equation}
where $f(\cdot)$ represents a chosen utility function. $f(\cdot)$ can be set as the number of clicks \citep{kveton2015cascading,lattimore2018toprank} or the number of correctly ordered pairs of documents \citep{jia2021pairrank}, based on the application context of OL2R. 

We have to emphasize that OL2R is \emph{\textbf{not}} literally just learning a ranking model online, such as using online gradient descent over a parametric model. The key challenge is the ranker only learns from the feedback collected in its presented rankings -- the so-called bandit feedback \citep{lattimore2020bandit}. For the results not presented to the users (i.e., presentation bias) or those presented at lower positions (i.e., position bias), the ranker has little knowledge about their relevance. As a result, the ranker has to repeatedly present the currently underestimated results on top for more feedback, which however contradicts its need to present the currently best ranking to reduce instantaneous regret. This is the well-known explore-exploit dilemma in OL2R. And the extremely large ranking space, i.e., exponentially sized possible rankings, further amplifies the difficulty in OL2R.
%The challenges in OL2R are two-fold: in addition to the well-known explore-exploit trade-off in all online learning problems with bandit feedback,  
Efficient exploration is thus needed to shrink the search space so as to quickly find satisfying result rankings to reduce regret. As discussed in Section \ref{sec_related}, various exploration strategies have been proposed in literature, with their own advantages and limitations. Our work in this paper is not to develop yet another OL2R solution, but to provide a general framework that ensures ranking fairness in a family of existing OL2R solutions. 

\subsection{Fairness in OL2R}
It is asserted by the probability ranking principle that user-side utility reaches its optimal when the documents are ranked by the expected values of their estimated relevance (merit) to the user \citep{robertson1977probability}. However, it unfortunately ignores the differential treatment that content providers would receive from the document side in such rankings.  
The key resource that a ranking system allocates among the documents is exposure \citep{singh2018fairness}, which influences the probability that the documents to be examined by the users and consequentially the social and/or economic benefit the content provider will receive. 
Without loss of generality, in this work, we assume all ranking candidates among all queries over time belong to two groups ($G_A$ and $G_B$). We define the instantaneous unfairness resulted in a particular ranked list as the difference between the exposure received by group $G_A$ and group $G_B$ from the presented ranking,
\begin{align*}
    UF(G_A, G_B, \beta) = Exposure(G_A) - \beta Exposure(G_B),
\end{align*}
where $\beta$ is an unfairness coefficient that controls the relative exposure that group $G_A$ and $G_B$ should receive. $\beta$ is an hyper-parameter chosen by the system designer to weigh the discrepancy in exposure between the two groups based on the need of specific applications. The definition of exposure is also application-specific: it can be quantified as the examination probability \citep{morik2020controlling} or the clicks/dwelling time received by the ranked documents. In this work, we adopt the classical choice of examination probability to simplify our discussion, but our solution can be generalized to other definitions with specific instrument to measure/infer the corresponding exposure.

To capture the difference of the cumulative treatment received during the course of OL2R up to time $t$, we consider a ranking system fair, if for any $t\in[T]$, we have
\begin{align}
\label{eqn:unfairness}
    UF_t = \left|\sum\nolimits_{s=1}^t\Big(Exposure_s(G_A) - \beta Exposure_s(G_B)\Big)\right|\le\epsilon,
\end{align}
where $\epsilon>0$ is a threshold chosen by the system designer to satisfy task-specific requirement of fairness. 

Our definition of unfairness is general to cover most of the existing exposure-based ranking fairness metrics. First, for demographic parity constraint proposed in \citep{singh2018fairness} which enforces that the average exposure of groups to be equal, $\beta$ can be set as the expected ratio of the group sizes between $G_A$ and $G_B$; while for fairness concerning the disparate treatment \citep{singh2018fairness, morik2020controlling}, of which the exposure received by each group should be proportional to their relevance or merit, $\beta$ can be set as the ratio of expected utility between groups. Besides, our unfairness definition is not limited to only two groups. Different pairs of groups can have different $\beta$ to control the relative exposure as long as the $\beta$s satisfy the transitivity constraint.

%% file: arxiv_method.tex
\section{Method}

In this section,  we present our proposed framework for fair OL2R. The key idea is to calibrate the trade-off between exploration and exploitation under fairness constraints. We choose PairRank \citep{jia2021pairrank} as the basic OL2R model because of its encouraging empirical performance and feasibility for  theoretical analysis. We use PairRank to illustrate our fair OL2R framework, and discuss how to generalize it with other OL2R algorithms that perform deterministic exploration in the action space. 

\subsection{Exploration and Exploitation in PairRank}
PairRank learns a pairwise logistic regression ranker on the fly. It scores a candidate ranking document $\xb$ by $f_{\btheta}(\xb)=\btheta^\top\bx$, in which the ranker's parameter $\btheta_t$ at round $t$ is obtained by optimizing the cross-entropy loss between the predicted pairwise relevance distribution on all documents and those inferred from user feedback till round $t$:
\begin{align}
\label{eqn:loss}
    \mathcal{L}_t (\btheta) = \sum\nolimits_{s=1}^t\sum\nolimits_{(m, n) \in D_s}  - (1 - \ymns)\log\big(1 - \sigma({\xmns}^\top\btheta)\big)   
    - \ymns \log\big(\sigma({\xmns}^\top\btheta)\big) + \frac{\lambda}{2} \Vert\btheta\Vert^2 
\end{align}
where $\xmns = \xb_m^s - \xb_n^s$, $D_s$ denotes the training set inferred from the click feedback at round $s$, $\ymns$ indicates the observed pairwise click preference between document $m$ and $n$, $\sigma(\cdot)$ is the sigmoid link function used in logistic regression, and $\lambda$ is the L2 regularization coefficient.

Due to click noise, there is uncertainty in the estimated parameter, i.e., $\Vert{\btheta}_t - \btheta^*\Vert \neq 0$, assuming $\btheta^*$ is the underlying ground-truth model parameter. Under the assumption that the pairwise click noise follows an $R$-sub-distribution (basically with finite variance), the uncertainty of PairRank's predicted pairwise rank order can be quantified and upper bounded with a high probability by the following lemma.

\begin{lemma} [Lemma 1 in \citep{jia2021pairrank}](Confidence Interval of Pairwise Rank Order Prediction in PairRank). At round $t < T$, for any pair of documents $\bx_i^t$ and $\bx_j^t$ under query $q_t$, with probability at least $1 - \delta_1$, we have,
\begin{equation*}
    |\sigma({\xijt}^\top {\btheta}_t) - \sigma({\xijt}^\top \btheta^*) | \leq \alpha_t\Vert\xijt\Vert_{\mathbf{M}_t^{-1}},
\end{equation*}
where $\alpha_t = ({2k_{\mu}}/{c_{\mu}}) \Big(\sqrt{R^2\log{({\det(\mathbf{M}_t)}/({\delta_1^2 \det(\lambda \mathbf{I})})})} + \sqrt{\lambda} Q\Big)$, $\mathbf{M}_t = \lambda \mathbf{I} + \sum_{s =1}^{t-1}\sum_{(m, n) \in D_s}\xmns\xmns^\top$, $k_{\mu}$ and $c_{\mu}$ are the Lipschitz constants of the sigmoid link function $\sigma(\cdot)$, $R$ is the sub-Gaussian parameter for the click noise, and $Q$ is the upper bound of the model parameter's norm, i.e., $\Vert\btheta\Vert \leq Q$.
\label{lemma:cb}
\end{lemma}

The key step in PairRank is to separate the document pairs into two sets: $S_t^c$ and $S_t^u$. The set $S_t^c$ contains all pairs where the two documents' relative rank order is certain (i.e., with high probability the predicted rank order is correct), and the complementary set $S_t^u$ consists of the pairs with uncertain rank order (i.e., their predicted rank order can still be wrong). And these two sets can be easily constructed using the results in Lemma \ref{lemma:cb}. 

When ranking the results, PairRank partitions all the candidate documents into blocks (by $S_t^c$ and $S_t^u$) such that all the pairwise ranking orders across blocks are certain (and thus correct with high probability), but there is still uncertainty among documents within a block. Then the algorithm outputs each block in a descending order (to exploit), while randomizes the rankings within each block (to explore). This divide-and-conquer strategy effectively reduces the exponentially large exploration space to quadratic.
And \citet{jia2021pairrank} proved that with a high probability PairRank's regret is only incurred when the rank orders in $S_t^c$ is not followed. This directly leads to the $\bigO{(\log^4(T))}$ cumulative regret upper bound defined on the number of mis-ordered pairs over $T$ rounds of result serving. However, no treatment was provided in PairRank to ensure fairness; and there is no guarantee its randomization leads to group fairness. In the next section, we describe our general receipt to turn PairRank, and the algorithms of its type, into a fair OL2R solution. 

\subsection{Fair Online Learning to Rank}
We now present our fair OL2R framework \model{} in detail. The key idea behind \model{} is to construct a set of group-level placement templates that calibrate exploration and exploitation with respect to the fairness constraint defined by Eq \eqref{eqn:unfairness} in every round.

\subsubsection{Fair placement of groups.}
Group fairness concerns whether each group of documents receive fair treatment, e.g., the exposure received by documents belonging to group $G_A$ should be comparable to that received by documents belonging to group $G_B$. Therefore, when considering the top-$k$ ranking positions, the placement of groups, instead of the detailed placement of documents, affects the unfairness across groups. Such a group-level ranking space for fairness control is much smaller than the exponentially sized ranking space for individual items. 
Specifically, with a fixed value $k$, which is chosen by the system designer for each specific ranking task, the set of all possible combinations of group rankings $\Omega^k$ is fixed and can be constructed beforehand. For example, with two groups at top-$k$ positions, there are in total $2^k$ different group-level placements, assuming each group has at least $k$ candidate items under each query.

For a particular group-level placement template for the top-$k$ positions, its resulting unfairness can be projected based on the expected exposure at each position, before the user examines the realized ranking. For example, under the position-based examination model \citep{craswell2008experimental}, the examination probability of each item $i$ only depends on its position in the ranked list, e.g., $P(e(i)=1 | i, \pi) = P(\pi(i))$. Several techniques have been proposed to estimate such position-based examination probabilities in practice \citep{agarwal2019estimating, fang2019intervention, joachims2017unbiased}. Other more sophisticated examination models can also be employed for the purpose to account for more factors than just the rank positions, such as document content \citep{wang2013content} and query intent \citep{yin2014exploiting}.
Then the expected exposure each group will receive can be estimated accordingly beforehand; and such values can be indexed for efficient online access later. 
For example, at round $t$, for the group placement template $\omega = \{A, A, B, A, B\} \in \Omega^5$, we can compute the expected exposure for group $G_A$ and $G_B$ as follows, 
\begin{align*}
    Exposure_t^{\omega}(G_A) = P(1) + P(2) + P(4), \quad Exposure_t^{\omega}(G_B) = P(3) + P(5).
\end{align*}
According to the definition in Eq \eqref{eqn:unfairness}, with the unfairness coefficient $\beta$, the expected cumulative unfairness if template $\omega$ is chosen for round $t$ can be estimated as,
\begin{align}
\label{eq_example_unfairness}
    UF_t^{\omega}(G_A, G_B, \beta) = UF_{t-1}(G_A, G_B, \beta) +P(1) + P(2) + P(4)- \beta(P(3) + P(5)) .
\end{align}
This quantity can then be used to guide the selection of final group placement template for round $t$. But we should note only after user examines the results we can compute the realized unfairness $UF_t(G_A, G_B, \beta)$ based on the actual exposure in the displayed ranking $\pi_t$, which will then be used to guide fairness control in the next round. Under specific types of exposure, e.g., those defined by position-based examination, the expected exposure can be directly used to update $UF_t(G_A, G_B, \beta)$.

\subsubsection{Calibrate explore-exploit trade-off under fairness constraint.}
Based on the notion of projected unfairness of a particular group placement template, we design our \model{} framework to ensure fairness during OL2R. 
Specifically, at round $t$, we first select all templates from $\Omega^k$ that satisfy Eq \eqref{eqn:unfairness} under the current cumulative unfairness $UF_{t-1}$ and the projected instantaneous unfairness induced by the templates. More formally, denote this set of templates as $\bar\Omega^k_t$, such that 
$\bar\Omega^k_t=\{\omega\in\Omega^k: |UF_{t-1} + Exposure_t^\omega(G_A) - \beta Exposure_t^\omega(G_B)|\le\epsilon, \}$.
For each group-level placement template $\omega$ in $\bar\Omega_t^k$, we use it to calibrate the exploration and exploitation at round $t$: the exploration will be confined within a subset of random permutations satisfying $\omega$; while for the exploitation, the fairness will be directly enforced (e.g., when the algorithm's chosen ranking is against the template's requirement, follow the template's).

\begin{algorithm}[t]
	\caption{\model{}-PairRank} 
	\label{alg:model} 
	\begin{algorithmic}[1]
    \STATE  \textbf{Input:} unfairness coefficient $\beta$, unfairness threshold $\epsilon$, L2 coefficient $\lambda$, uncertainty coefficient $\delta$, length of the returned ranked list $k$.
    \STATE Initialize $\mathbf{M}_0 = \lambda \mathbf{I},  {\btheta}_1 = 0, UF_0 = 0.$
	\FOR{$t=1, \dots, T$}
	\STATE $q_t \leftarrow $receive\_query($t$) 
	\STATE $\cX_t = \{\xb_1^t, \cdots, \xb_{L_t}^t\} \leftarrow $ retrieve\_documents($q_t$)
	\STATE $\Omega_t \leftarrow$ retrieve\_group\_placement($k$, $\epsilon$, $\beta$, $UF_{t-1}$)
	\STATE $S_t^c, S_t^u \leftarrow $ construct\_order\_sets($\cX_t, \btheta_{t-1}, \Mb_{t-1}, \delta$)
    \STATE $\pi_t \leftarrow$ FairSwap($S_t^c, S_t^u, \Omega_t$) 
	\STATE $C_t \leftarrow $ collect\_click\_feedback($\pi_t$)
	\STATE $UF_t \leftarrow$ $UF_{t-1}$ + unfairness($\pi_t, C_t$)
	\STATE $D_t \leftarrow $ construct\_training\_set($\pi_t, C_t$)
    \STATE Obtain $\btheta_t$ by minimizing Eq \eqref{eqn:loss}\\
    \STATE $\Mb_{t} = \Mb_{t-1} + \sum_{(i, j)\in D_t} \xb_{ij}{\xb_{ij}}{^\top}$
	\ENDFOR
	\end{algorithmic}
\end{algorithm}

To make our discussion more concrete, we demonstrate how to cast PairRank into a fair OL2R solution with our proposed \model{} framework in Algorithm \ref{alg:model}. The key modification of PairRank lies in step 8 that generates the ranked list with respect to the qualified templates. 
To find the satisfying fair ranking at each round $t$ in PairRank, the most straightforward way is to first generate all ranked lists based on PairRank's original ranking strategy that satisfy the templates in $\bar\Omega^k_t$, and then return the one with minimum expected regret. As mentioned before, the expected regret can be readily computed based on the violation of certain rank order set $S_t^c$ in PairRank \citep{jia2021pairrank}; but the search space is exponentially large: any random permutation of documents in the same block is a valid ranking in PairRank, because the rank orders among those documents are still uncertain and PairRank needs such random permutations to obtain unbiased feedback.  
Hence though valid, such an exhaustive search is prohibitively expensive for online result serving. 

In FairExp-PairRank, we propose to modify the block structure in PairRank to efficiently find the ranking with minimum expected regret under each qualifying template, and then locate the best ranking across all templates in $\bar\Omega^k_t$. 
More specifically, PairRank only requires the relative order among blocks to be followed; while within each block, any permutation of associated documents is allowed. Then for a chosen template, we can segment it according to the size of each block that is already ranked in the descending order, and then calibrate the blocks segment by segment. 
For the segment where we do not have enough documents for a group in the corresponding block, we have to fetch satisfying documents from the lower ranked block(s), which introduces added regret on top of PairRank's original regret. 
But for segments we have enough number of documents in each group, no modification is needed in the corresponding block and thus no added regret. 
Therefore, we just need to guarantee the way we modify the blocks segment by segment will induce the minimum added regret.

We propose FairSwap for this purpose. It takes a group-level placement template and the original partition from PairRank as input, and outputs a qualified ranking and corresponding added regret. 
We use an intuitive example in Figure~\ref{fig:fairswap} to illustrate the design of FairSwap. In this example, PairRank already partitions the five documents into two blocks: $\{\xb_1, \xb_2\}$ and $\{\xb_3, \xb_4, \xb_5\}$, and requires the first block to be ranked above the second. The chosen template $\omega$ requires the ranking to follow $\{A, A, B, A, B\}$. Then, FairSwap first segments $\omega$ into two parts $\{A, A\}$ and $\{B, A, B\}$ based on the size of the two blocks, and examines them one by one in order. 
For the first segment, it requires two documents from $G_A$, but there is only one such document in the first block. Hence FairSwap needs to fetch one from the next block. As PairRank does not specify whether $\xb_3$ is better than $\xb_4$ or vice versus, FairSwap will arbitrarily choose one of them into the first block to avoid any presentation bias. But because the first segment only requires two documents and they must belong to $G_A$, $\xb_2$ has to be moved downwards. To reflect the fact that $\xb_2$ is guaranteed to be better than any documents in the current second block, FairSwap creates a new block only containing $\xb_2$ and inserts it before the next block. Correspondingly, the next segment is also further segmented to accommodate this new block. Then FairSwap moves to the next segments, until new violations emerge between a segment and its corresponding block. Then the same procedure will be applied to calibrate the block structure. Once all the segments are satisfied, an arbitrary ranking will be generated with respect to the new partition just as in PairRank. In this result, because $\xb_4$ is moved up ahead of $\xb_2$ to satisfy $\omega$, added regret is introduced. And this added regret can be precisely computed using the certain rank order set $S_t^c$ in PairRank. For instance, in our example, the added regret will be 2, if the finally output ranking is $\{\xb_4 \succ \xb_1 \succ \xb_2 \succ \xb_3 \succ \xb_5\}$.

\begin{figure}[t]
    \centering
    \includegraphics[width=0.8\linewidth]{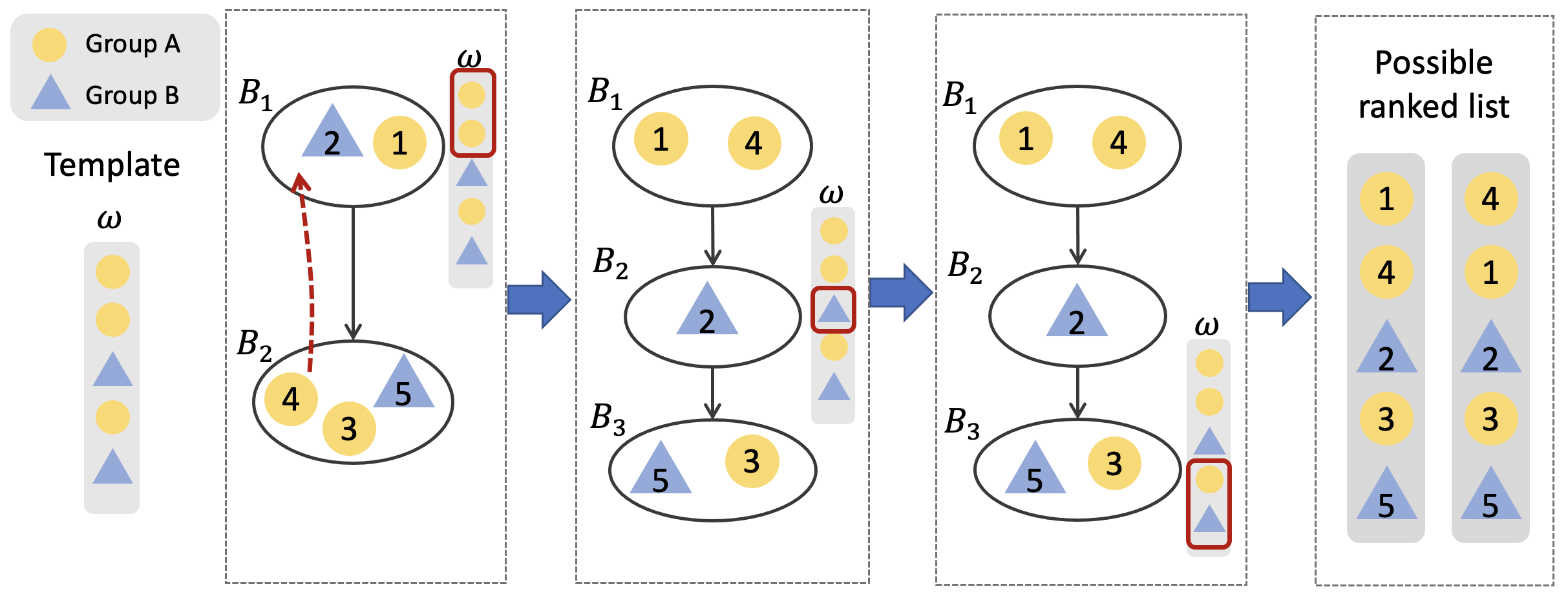}
    \caption{Illustration of FairSwap: calibrate pairwise exploration in PairRank under fairness constraint. It satisfies the required fair placement of groups by minimum number of swaps between blocks generated in PairRank.}
    \label{fig:fairswap}
    \vspace{-3mm}
\end{figure}

The next theorem states FairSwap finds a satisfying ranking for a given template with minimum added regret.
\begin{theorem}
Under the setting of PairRank, to satisfy template $\omega$, FairSwap can generate the minimum added regret.
\end{theorem}
\begin{proof}
At a particular round, assume $L$ blocks are constructed for the documents associated with query $q$ by PairRank. For each block $B_i$, there are $n_i$ documents in $G_B$, $i \in [L]$. Without loss of generality, suppose according to template $\omega$, block $B_b$ needs extra $M$ documents from $G_A$; and all the segments before this has been satisfied. To satisfy this, we need to move
$M = \sum\nolimits_{i=b+1}^L m_i$ documents from lower ranked blocks into $B_b$, where $m_i$ represents the number of documents swapped from block $i$ to block $b$; when no document is needed from block $i$ to move up, $m_i=0$.

Because those promoted documents are known to be worse than those from higher ranked blocks in the original partition, the added regret caused by such swaps can be calculated by the product between the number of promoted documents and the number of those being moved over:
\begin{align*}
    r^{added} \le \sum\nolimits_{i=b+1}^L m_i \times \left(M + \sum\nolimits_{j=b+1}^{i-1} n_j\right) 
    = M^2 + \sum\nolimits_{i=b+1}^{L}m_i \sum\nolimits_{j=b+1}^{i-1} n_j.
\end{align*}
The inequality is because once the documents from different blocks (e.g., $B_i$ and $B_j$ where $i<j$) are moved into one block, the random shuffling in PairRank will disregard their original known order (presumably documents from $B_i$ should still be ranked above those from $B_j$).
$M$ is determined by the required template $\omega$ and the original partition from PairRank after block $B_b$; and hence no swapping can reduce it. In the second term of right-hand side, as $n_j$ is non-zero, $\sum\nolimits_{j=b+1}^{i-1} n_j$ increases with $i$. Therefore, a greedy method for selecting $m_j$ from the closest next block will lead to the minimum added regret in PairRank.
\end{proof}

Since FairSwap is guaranteed to find the satisfying ranking with the minimum added regret for a given template, we can execute it against all templates in $\bar\Omega^k_t$ and return the one with minimum added regret for output. As all templates in $\bar\Omega^k_t$ are qualified to satisfy Eq \eqref{eqn:unfairness}, the final output is also satisfied and has minimum added regret for PairRank.

\subsubsection{Discussions} As the authors suggested in the original paper of PairRank, when generating permutations in each block, certain rank orders in $S^c_t$ can be followed to empirically improve its online ranking quality, though no theoretical justification is provided \citep{jia2021pairrank}. This heuristic also applies when we perform FairSwap: when moving documents from a lower ranked block up, we should choose those who are known to be better than many other documents in the same block. For example, in Figure \ref{fig:fairswap}, if we know $\xb_3\succ\xb_4$ (but not their relations with $\xb_5$), we should move $\xb_3$ to the upper block, rather than $\xb_4$, to void the misordering between $\xb_3$ and $\xb_4$. In our experiments, we empirically verified this strategy is indeed helpful.

So far, all our discussions about \model{} are based on PairRank; but our principle for obtaining fair OL2R can be readily applied to a series of solutions, where the only requirement is their exploration can be explicitly computed (e.g., those confidence interval based solutions \citep{lattimore2018toprank,li2018online,li2016contextual,kveton2015cascading}). The basic idea is clear: at each round $t$, we first construct the set of qualified group-level placement templates $\bar\Omega^k_t$, which are independent from the specific OL2R solutions. Then under each template we find all satisfying rankings from the chosen OL2R solution, and then choose the one with the minimum added regret. For example, in TopRank \citep{lattimore2018toprank}, we can use the template to restrict their confidence interval based topological sort; and in RecurRank \citep{li2018online} the template can be used to confine their exploration in each position. The research question is then how to design an efficient algorithm to take advantage of target algorithm's output structure, such as FairSwap on PairRank's partition structure. This is clearly algorithm-specific, and we leave it as our future work.

%% file: arxiv_experiment.tex
\section{Experiment}

In this section, we empirically compared our proposed fair OL2R solution with existing state-of-the-art OL2R and fair OL2R algorithms on two public learning to rank benchmark datasets. 

\noindent{\bf Reproducibility}
Our entire codebase, baselines, analysis, and experiments can be found on
Github~\footnote{https://github.com/yilingjia/FairExp}.

\subsection{Experiment Setup}

\begin{table}[]
\centering
\caption{Group information for two datasets.}
  \vspace{-2mm}
  \label{table:group}
\begin{tabular}{llllll}
\hline
\textbf{Dataset}                 & \textbf{Group}         & \multicolumn{2}{l}{Group A}            & \multicolumn{2}{l}{Group B} \\ \cline{3-6} 
                        & \textbf{attribute}     & Size    & \multicolumn{1}{l|}{Utility} & Size         & Utility      \\ \hline
\textbf{MSLR}                    & PageRank      & 50.57\% & \multicolumn{1}{l|}{0.7302}  & 49.43\%      & 0.6093      \\
\textbf{WEB10K}                  & Inlink number & 64.85\% & \multicolumn{1}{l|}{0.6665}  & 35.15\%      & 0.6776        \\ \hline
\multirow{2}{*}{\textbf{Yahoo!}} & Feature 9     & 49.22\% & \multicolumn{1}{l|}{1.2371}   & 50.78\%      & 0.5893       \\
                        & Feature 471   & 85.75\% & \multicolumn{1}{l|}{1.2384}  & 14.25\%      & 1.2239       \\ \hline
\end{tabular}
\vspace{-2mm}
\end{table}

\subsubsection{Datasets} We experimented on two publicly available learning to rank datasets, Yahoo! Learning to Rank Challenge dataset \citep{chapelle2011yahoo}, which consists of 292,921 queries and 709,877 documents represented by 700 ranking features, and MSLR-WEB10K \citep{qin2013introducing}, which contains 10,000 queries, each having 125 documents on average represented by 136 ranking features. Both datasets are labeled on a five-grade relevance scale: from not relevant (0) to perfectly relevant (4). We followed the train/test/validation split provided in the datasets to perform the cross-validation.

Since no group information about the documents is provided in these two datasets, we manually divided the datasets into two groups. For MSLR-WEB10K dataset, where the detailed information about each ranking feature is provided by the data publisher, we chose the document-specific feature PageRank (feature id 130) and inlink number (feature id 128) as the group attribute. PageRank separates the datatset into two groups with similar size, while inlink number roughly separates the dataset into two groups with similar average utility. But the Yahoo dataset did not provide any specific feature information. To follow a similar split as obtained in MSLR-WEB10K dataset (i.e., size vs., average utility), we checked all the features and chose feature 9 and feature 471 as the group attribute. The detailed group constructions in these two datasets are provided in Table~\ref{table:group}. In addition to studying the trade-off among exploration, exploitation and fairness in OL2R, such choice of group attributes will demonstrate how the intrinsic property of the problem/dataset affects these three elements in fair OL2R.

\subsubsection{User interaction simulation} User clicks were simulated via the standard procedure for OL2R evaluations \citep{oosterhuis2018differentiable, jia2021pairrank}. At each round, a query is uniformly sampled from the training set for result serving. A ranked list will be determined by the model and returned to the user. User clicks are simulated with a dependent click model (DCM)~\citep{guo2009efficient}, which assumes that the user will sequentially scan the list and make click decisions on the examined documents. In DCM, the probabilities of clicking on a given document and stopping the subsequent examination are both conditioned on the document's true relevance label, e.g, $\mathbb{P}(click = 1 | \text{relevance grade})$, and $\mathbb{P}(stop = 1 |click = 1, \text{relevance grade})$.

We employ three different click model configurations to represent three different types of users, for which the details are shown in Table \ref{table:click}. Basically, we have the \textit{perfect} users, who click on all relevant documents and do not stop browsing until the last returned document; the \textit{navigational} users, who are very likely to click on the first highly relevant document and stop there; and the \textit{informational} users, who tend to examine more documents, but sometimes click on irrelevant ones, and thus contribute a significant amount of noise in their click feedback.
To reflect presentation bias, only top $k=10$ ranked results are returned to the users. If there is no click, the user will continue examining the next position. 

\begin{table}[!htbp]
\vspace{-2mm}
  \caption{Configuration of simulated click models.}
  \vspace{-2mm}
  \label{table:click}
  \centering
  
  \begin{tabular}{cccccc|ccccc}
    \hline
                & \multicolumn{5}{c}{Click Probability} & \multicolumn{5}{c}{Stop Probability} \\
R & 0           & 1          & 2   &3 &4        & 0          & 1          & 2   &3 &4       \\ \hline
\textit{per}         & 0.0         & 0.2        & 0.4 &0.8 &1.0        & 0.0        & 0.0        & 0.0   &0.0 &0.0     \\
\textit{nav}    & 0.05   &0.3      & 0.5    &0.7    & 0.95       & 0.2       &0.3 & 0.5    &0.7    & 0.9        \\
\textit{inf}   & 0.4      &0.6   & 0.7   &0.8     & 0.9        & 0.1      &0.2  & 0.3     &0.4   & 0.5        \\ \hline
\end{tabular}
\vspace{-2mm}
\end{table}

\subsubsection{Baselines and settings} We compared the proposed fair framework under PairRank, with several state-of-the-art OL2R solutions, and the fair OL2R model, FairCo. The details are list below.
\begin{itemize}\setlength\itemsep{0.05em}
    \item \textbf{DBGD} \citep{yue2009interactively}: DBGD uniformly samples a direction from the entire model space for exploration and model update.
    \item \textbf{PDGD} \citep{oosterhuis2018differentiable}: 
    PDGD samples the next ranked document from a Plackett-Luce model and estimates gradients from the inferred pairwise preferences.
    \item \textbf{PairRank} \citep{jia2021pairrank}: PairRank directly learns a pairwise model with logistic regression, and explores based on the model's uncertainty with divide-and-conquer.
    \item \textbf{FairCo}\citep{morik2020controlling}: FairCo achieves the group fairness with a proportional control based on the exposure of each group.
\end{itemize}

For each fold, the models are trained with the simulated clicks on the training dataset, and the hyper-parameters are selected based on their offline performance (relevance learning) on the validation set. We used grid search for the best set of hyper-parameters. For PDGD and DBGD, learning rate is selected among $\{10^{-i}\}_{i=1}^3$. For PairRank, we did a grid search for its regularization parameter $\lambda$ over $\{10^{-i}\}_{i=1}^3$ and exploration parameter $\alpha$ over $\{10^{-i}\}_{i=1}^3$. FairCo depends on the pre-defined examination probability of each position to perform inverse propensity scoring. We followed \citep{joachims2017ips} to use randomization to estimate the position-based examination probabilities. The fairness coefficient in FairCo is searched over $\{10^{-i}\}_{i=1}^3$. For FairExp-PairRank, we did the same grid search for the hyper-parameters associated with PairRank. 
To make the fairness control comparable between FairExp-PairRank and FairCo, we set $\beta$ to the ratio of the true utility of each group. For example, for MSLR-WEB10K dataset with inlink as the group attribute, we set $\beta$ to 1. For FairCo, we set the merit of each group based on the ground-truth relevance labels. To calculate exposure, both FairExp-PairRank and FariCo use the independently estimated examination probabilities as the exposure received by each document. With such a setting, both FairCo and FairExp-PairRank focus on the merit-based exposure for group fairness.

\subsubsection{Evaluation} We evaluate the fair OL2R models under three aspects, user satisfaction (i.e., exploitation), relevance learning (i.e., exploration), and unfairness treatments across groups over time. 

For user satisfaction, we evaluate the models' ranking performance during online result serving. In fair OL2R, users' experience will be traded off with the purposes of the feedback collection and the fairness control. We adopt the cumulative NDCG over $T$ rounds to assess models' online performance,
\begin{equation*}
    \text{Cumulative NDCG} = \sum\nolimits_{t=1}^T \text{NDCG@10}(\tau_t) \cdot \gamma^{(t-1)},
\end{equation*}
which computes the expected utility a user receives with probability $\gamma$ that the user stops searching after each query~\citep{oosterhuis2018differentiable, jia2021pairrank}. And in our evaluation, we set $\gamma = 0.9995$. 
For relevance learning, at each round, the newly updated ranker will be evaluated on a hold-out testing set against its ground-truth relevance label. This measures how fast the OL2R model improves its ranking quality, which is measured by NDCG@10 on this hold-out set. Besides, we calculate unfairness according to Eq~\eqref{eqn:unfairness} with a chosen $\beta$ value. 

\begin{figure*}[t]
  \centering
  \begin{subfigure}[b]{\textwidth}
    \centering
    \includegraphics[width=\linewidth]{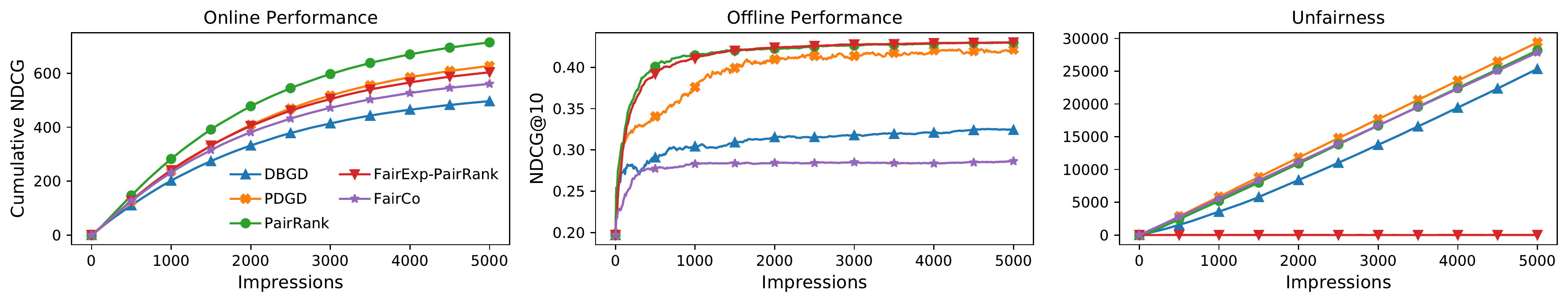}
    \caption{MSLR-WEB10K dataset with group attribute=PageRank score, $\beta=1.22$}
    \label{fig:pagerank_per}
  \end{subfigure}
  \begin{subfigure}[b]{\textwidth}
    \centering
    \includegraphics[width=\linewidth]{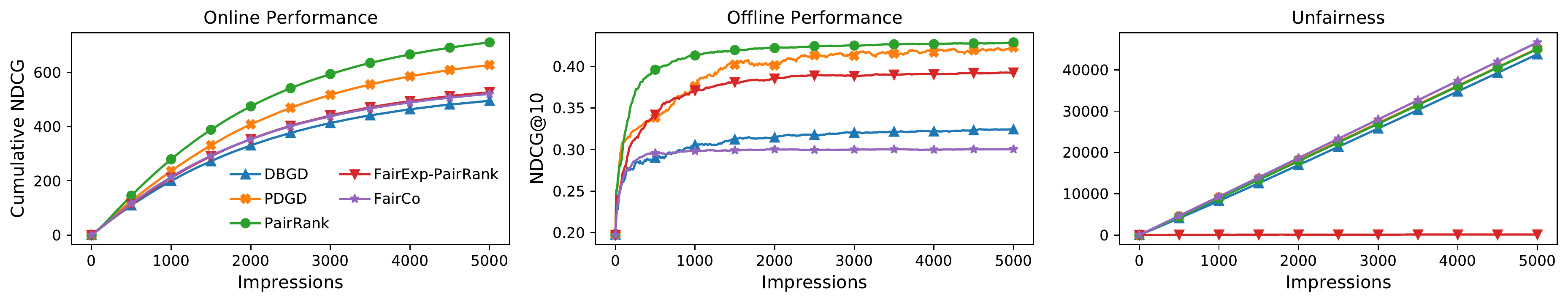}
    \caption{MSLR-WEB10K dataset with group attribute=number of inlinks, $\beta=1.0$}
    \label{fig:inlink_per}
  \end{subfigure}
  \caption{Online, offline and fairness performance for MSLR-WEB10K under perfect click model.}
  \label{fig:MSLR_per}
\end{figure*}

\begin{figure*}[t]
  \centering
  \begin{subfigure}[b]{\textwidth}
    \centering
    \includegraphics[width=\linewidth]{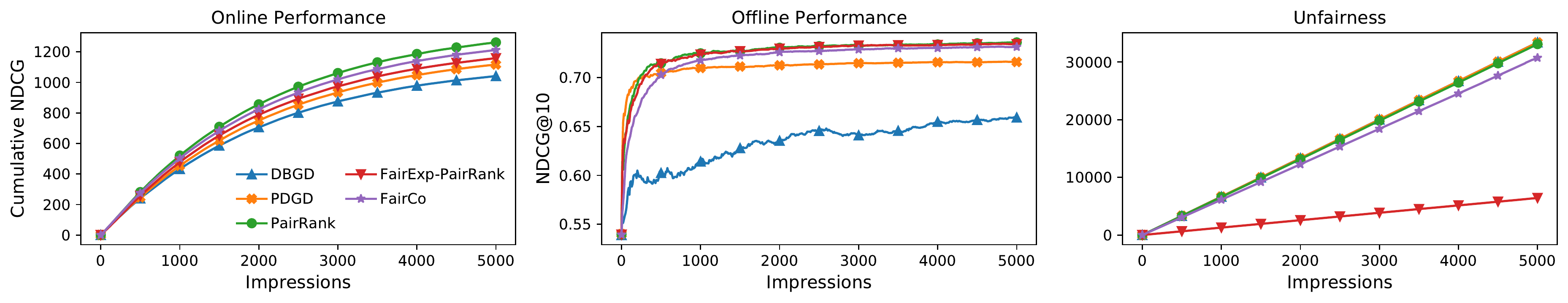}
   \caption{Yahoo dataset with group attribute=feature 9, $\beta=1.0$}
    \label{fig:9_per}
  \end{subfigure}
  \begin{subfigure}[b]{\textwidth}
    \centering
    \includegraphics[width=\linewidth]{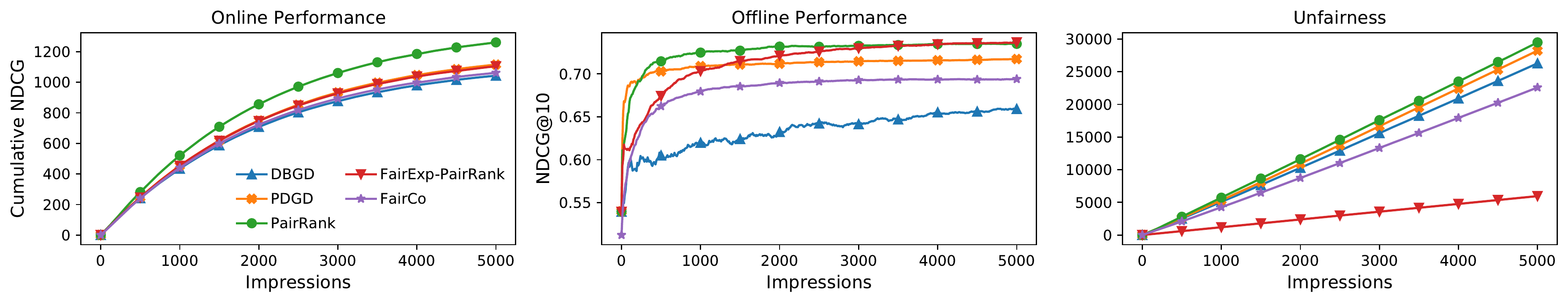}
    \caption{Yahoo dataset with group attribute=feature 471, $\beta=0.477$}
    \label{fig:471_per}
  \end{subfigure}
  \caption{Online, offline and fairness performance for Yahoo dataset under perfect click model.}
  \label{fig:yahoo_per}
\end{figure*}

\begin{figure*}[t]
  \centering
  \begin{subfigure}[b]{\textwidth}
    \centering
    \includegraphics[width=\linewidth]{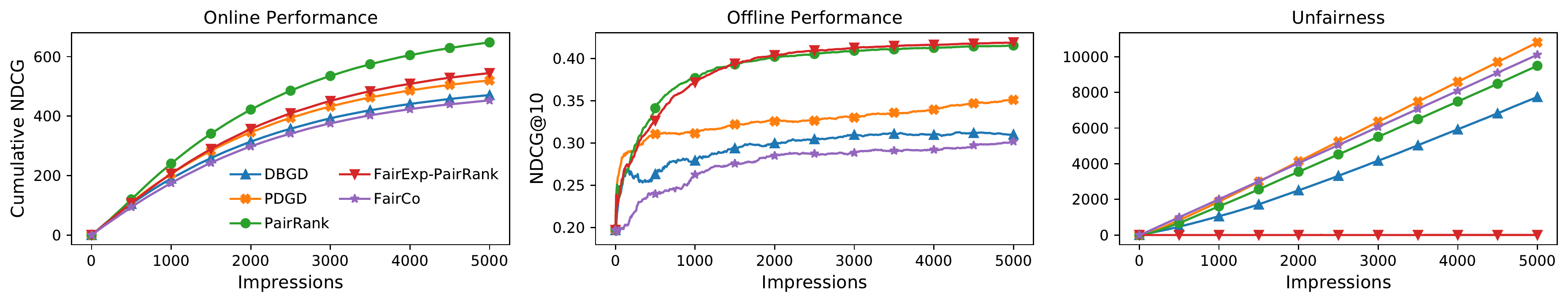}
    \caption{MSLR-WEB10K dataset with group attribute=PageRank score, $\beta=1.22$}
    \label{fig:pagerank_nav}
  \end{subfigure}
  \begin{subfigure}[b]{\textwidth}
    \centering
    \includegraphics[width=\linewidth]{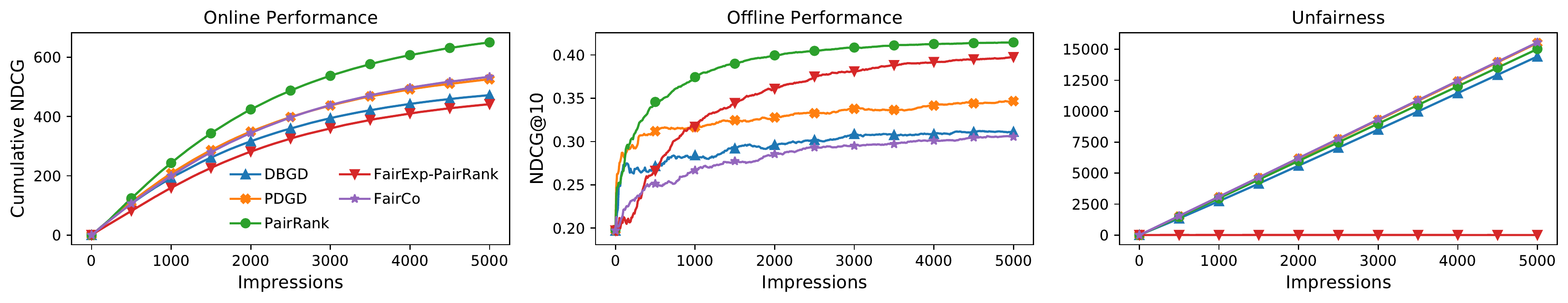}
    \caption{MSLR-WEB10K dataset with group attribute=number of inlinks, $\beta=1.0$}
    \label{fig:inlink_nav}
  \end{subfigure}
  \caption{Online, offline and fairness performance for MSLR-WEB10K under navigational click model.}
  \label{fig:MSLR_nav}
\end{figure*}

\begin{figure*}[t]
  \centering
  \begin{subfigure}[b]{\textwidth}
    \centering
    \includegraphics[width=\linewidth]{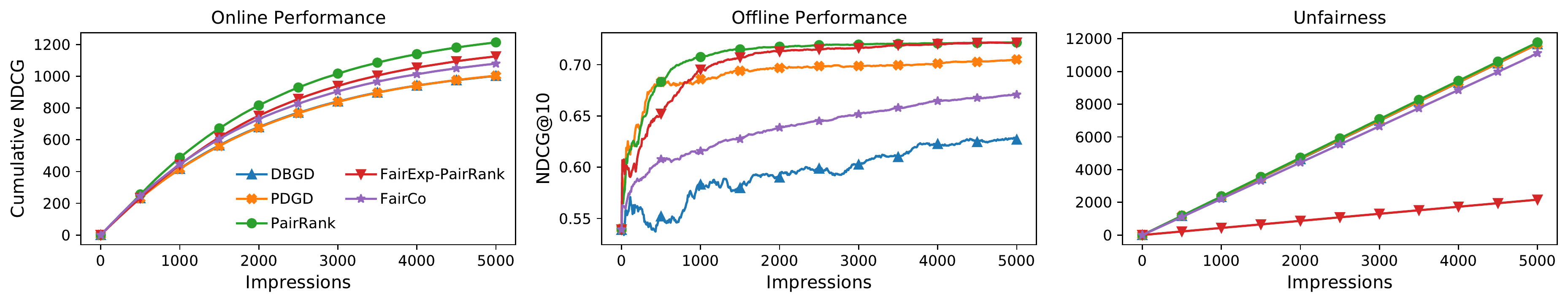}
   \caption{Yahoo dataset with group attribute=feature 9, $\beta=1.0$}
    \label{fig:9_nav}
  \end{subfigure}
  \begin{subfigure}[b]{\textwidth}
    \centering
    \includegraphics[width=\linewidth]{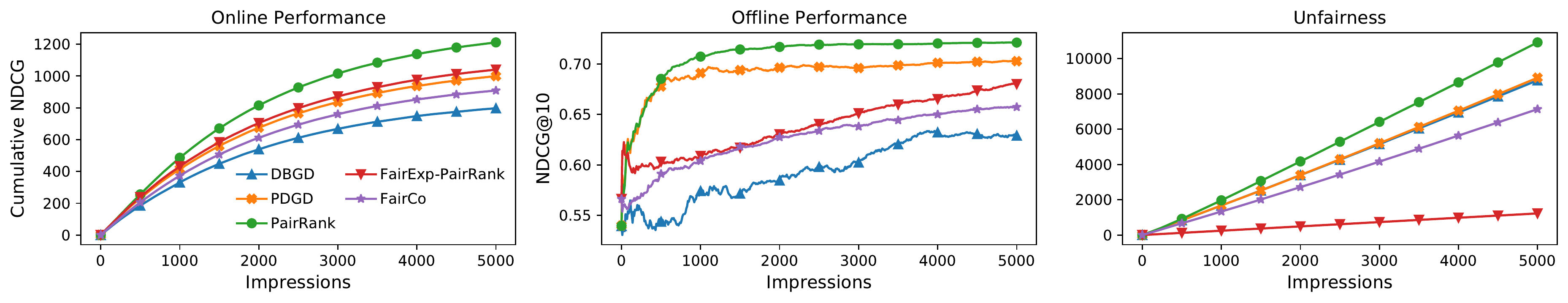}
    \caption{Yahoo dataset with group attribute=feature 471, $\beta=0.477$}
    \label{fig:471_nav}
  \end{subfigure}
  \caption{Online, offline and fairness performance for Yahoo dataset under navigational click model.}
  \label{fig:yahoo_nav}
\end{figure*}

\begin{figure*}[t]
  \centering
  \begin{subfigure}[b]{\textwidth}
    \centering
    \includegraphics[width=\linewidth]{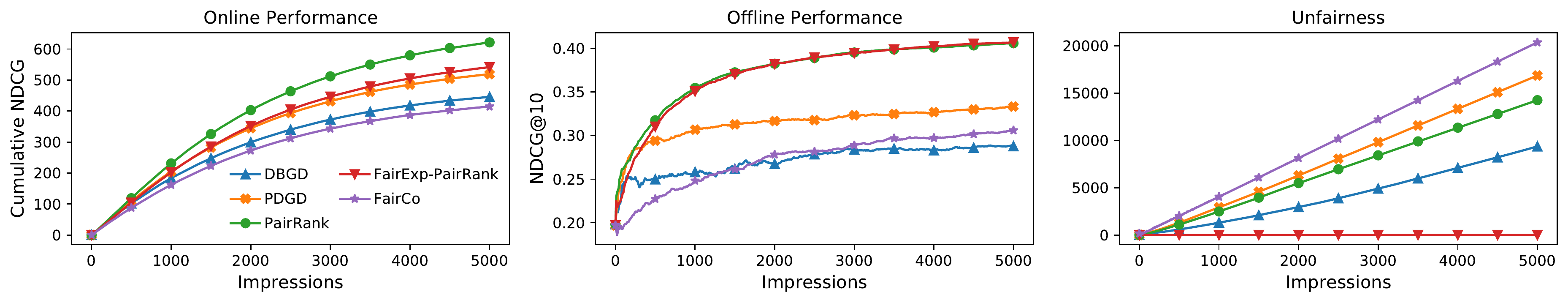}
    \caption{MSLR-WEB10K dataset with group attribute=PageRank score, $\beta=1.22$}
    \label{fig:pagerank_inf}
  \end{subfigure}
  \begin{subfigure}[b]{\textwidth}
    \centering
    \includegraphics[width=\linewidth]{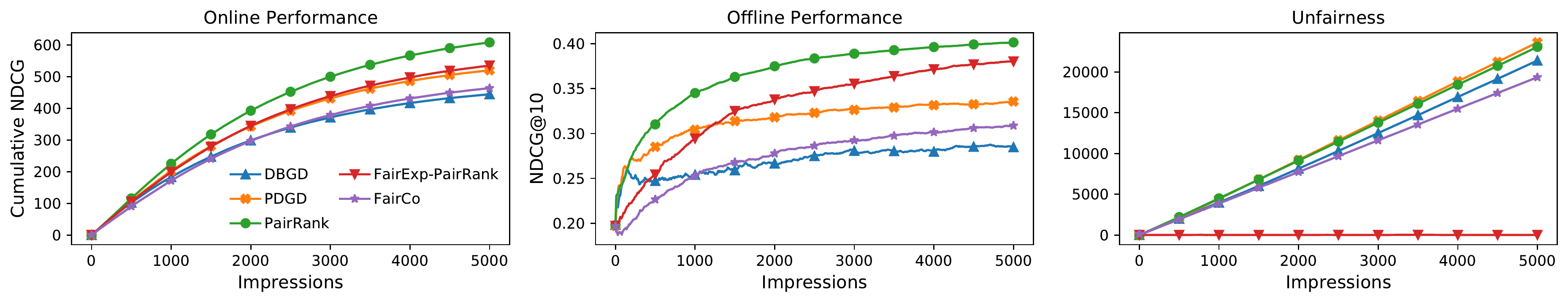}
    \caption{MSLR-WEB10K dataset with group attribute=number of inlinks, $\beta=1.0$}
    \label{fig:inlink_inf}
  \end{subfigure}
  \caption{Online, offline and fairness performance for MSLR-WEB10K under informational click model.}
  \label{fig:MSLR_inf}
\end{figure*}

\begin{figure*}[t]
  \centering
  \begin{subfigure}[b]{\textwidth}
    \centering
    \includegraphics[width=\linewidth]{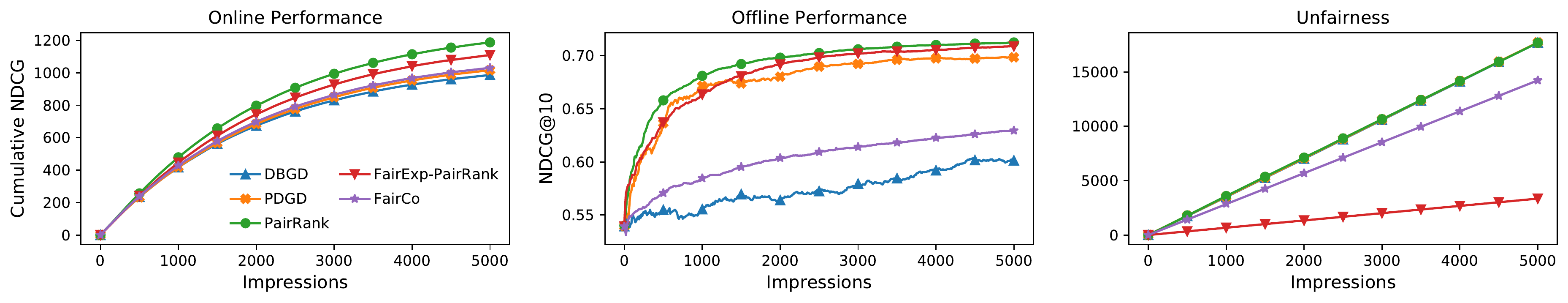}
   \caption{Yahoo dataset with group attribute=feature 9, $\beta=1.0$}
    \label{fig:9_inf}
  \end{subfigure}
  \begin{subfigure}[b]{\textwidth}
    \centering
    \includegraphics[width=\linewidth]{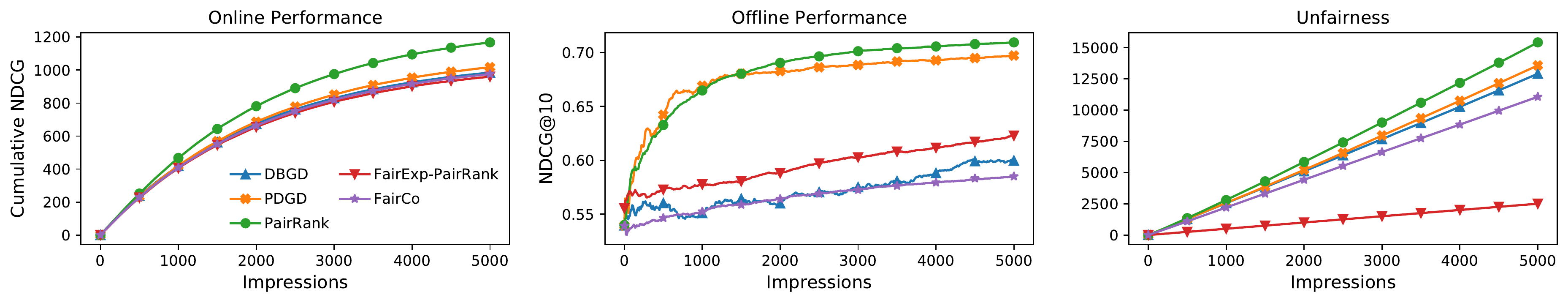}
    \caption{Yahoo dataset with group attribute=feature 471, $\beta=0.477$}
    \label{fig:470_inf}
  \end{subfigure}
  \caption{Online, offline and fairness performance for Yahoo dataset under informational click model.}
  \label{fig:yahoo_inf}
\end{figure*}

\subsection{Experiment Results}

\subsubsection{Exploration, exploitation, and fairness} We report the online performance, offline performance and unfairness on the two datasets with four different group attributes in Figure \ref{fig:MSLR_per} to \ref{fig:yahoo_inf}. As we have three different click models and two datasets, and under each datasets we have two different ways to define groups, we will have 12 different combinations for evaluations. To facilitate our discussion, we organize the results first by click models, then by datasets, and finally by the group definition. This gives us a more organized picture of all our comparisons.

First, we can clearly observe that PairRank consistently showed the best online performance across all the settings. And due to the fairness control, FairExp-PairRank encounters reduced online performance. This is the expected trade-off between utility and fairness: 
to allocate fair exposure for different groups, documents with low-quality have to be swapped to the top positions. Fortunately, due to our FairSwap strategy that guarantees minimum added regret and the already good performance of PairRank, the cost of such trade-off is under control, such that FairExp-PairRank even outperformed most of other OL2R solutions that are not subject to fairness constraints. On the other hand, FairCo showed worse online performance than FairExp-PairRank and the other OL2R solutions. We attribute this to FairCo's ``overreacting'' fairness control. In FairCo, to reduce the current unfairness, the ranking scores of documents belonging to the underrepresented group will all be promoted by the same value; but the differences of their relevance quality are ignored. As a result, documents with low-quality will be indifferently promoted for fairness, which directly results in its bad online performance.

Offline performance indicates the models' convergence for relevance learning. We can observe that the convergence rate of FairExp-PairRank is slower than that in PairRank. This is consistent with our previous discussion that due to the fairness control some documents originally selected for exploration cannot be displayed, which directly slows down the improvement of relevance estimation. 
Besides, we can observe that such control shows different impact across different group settings. For example, there is almost no impact on the offline performance of FairExp-PairRank on MSLR-WEB10K dataset when choosing PageRank feature to define groups, while on Yahoo dataset with feature 471 as the group attribute, the drop is significant. We checked the detailed output in FairExp-PairRank and found that in the former case, the fairness constraint can be largely satisfied with the original rankings in PairRank. Therefore, little calibration is added. However, in the latter case the underexplored documents always got demoted by fairness control (because the group size differs significantly), it seriously slowed down relevance learning. For example, on Yahoo dataset, when choosing the feature 471 to define the groups, one group has almost 5 times more documents than another group. As a result, FairExp-PairRank has to repeatedly promote the minority for fairness, i.e., restrain the documents belonging to the majority group from the top positions.  

On the other hand, thanks to our template-guided fairness control, FairExp-PairRank showed significant advantages in its cumulative unfairness compared to all baselines across all settings, as fairness is handled as a hard constraint in it. 
It is worth noting that on the Yahoo dataset, even with FairExp-PairRank, the cumulative unfairness is increasing. This is because there are many queries with highly imbalanced candidate documents from the two groups such that no ranking can be generated with respect to the fair templates in $\bar\Omega$. In FairExp, when no qualified template can be satisfied, we will first choose a ranking with the minimum unfairness, and then compare their added regret. And this observation also reflects the intrinsic incompatibility of an environment in fair ranking problems: the distribution of queries and/or their associated documents is out of the control of an algorithm. Compared to  FairExp-PairRank that treats unfairness as a hard constraint, FairCo shows much higher unfairness during the course of OL2R. And this is still resulted from its ``overreacting'' behavior in the proportional control: it aims to fix all the previous unfair treatments in the current query, which might over compensates the current underrepresented group and lead to oscillation in the subsequent controls. 

In addition, as we mentioned before, even when pure random exploration is employed in OL2R, e.g., in DBGD, ranking fairness is not automatically satisfied. Especially when the optimal relevance-driven ranking is unfair, the faster an OL2R converges, the more unfair its ranked results will be.

\begin{figure*}[t]
  \centering
  \begin{subfigure}[b]{\textwidth}
    \centering
    \includegraphics[width=\linewidth]{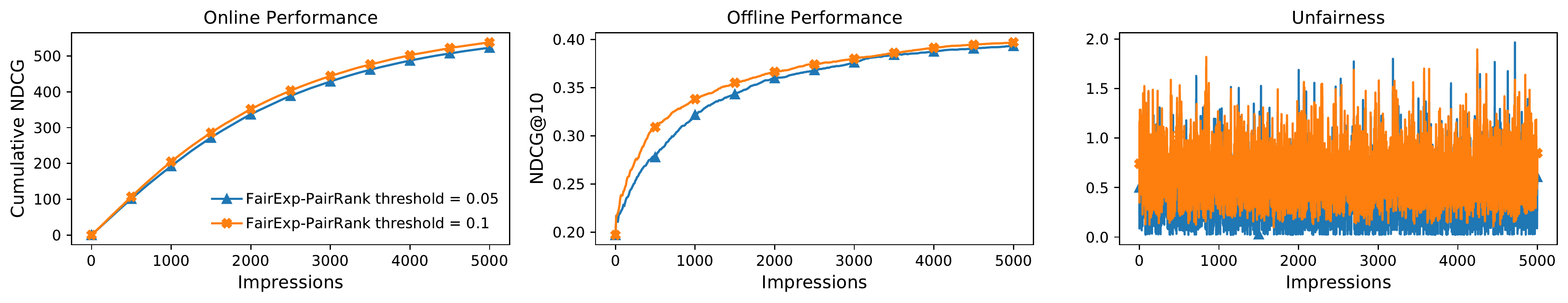}
     \caption{MSLR-WEB10K dataset with group attribute=PageRank score, $\beta=1.22$}
    \label{fig:controlpagerank}
  \end{subfigure}
  \begin{subfigure}[b]{\textwidth}
    \centering
    \includegraphics[width=\linewidth]{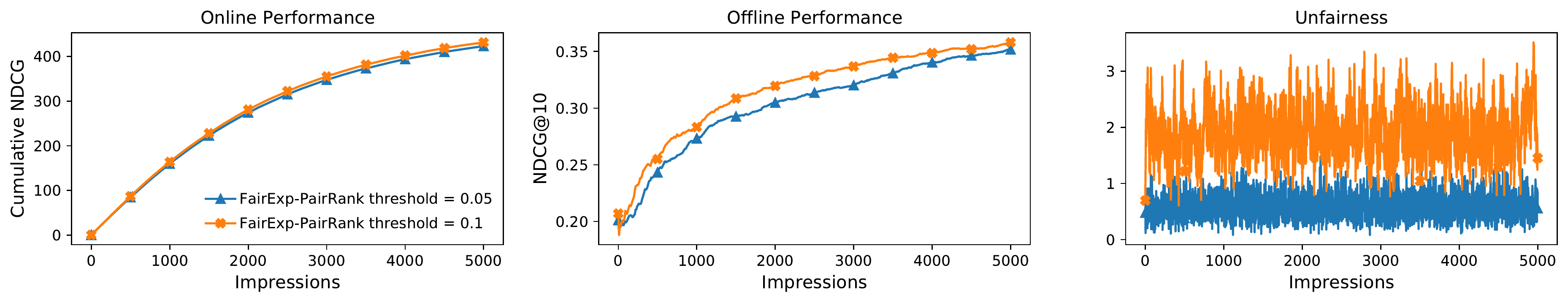}
    \caption{MSLR-WEB10K dataset with group attribute=number of inlinks, $\beta=1.0$}
    \label{fig:controlinlink}
  \end{subfigure}
\caption{Effect of fairness control on online and offline performance for MSLR-WEB10K dataset under informational click model.}
  \label{fig:web10kcontrol}
\end{figure*}

\begin{figure*}[t]
  \centering
  \begin{subfigure}[b]{\textwidth}
    \centering
    \includegraphics[width=\linewidth]{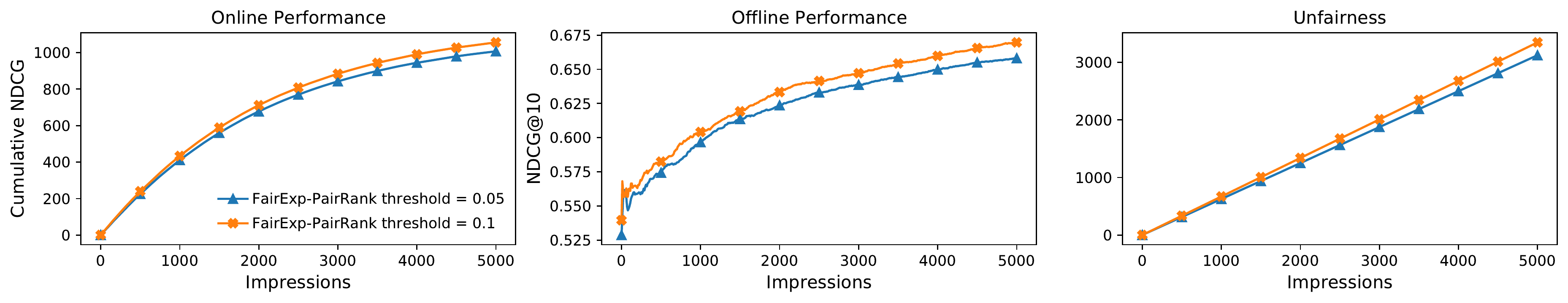}
     \caption{Yahoo dataset with group attribute=feature 9, $\beta=1.0$}
    \label{fig:control9}
  \end{subfigure}
  \begin{subfigure}[b]{\textwidth}
    \centering
    \includegraphics[width=\linewidth]{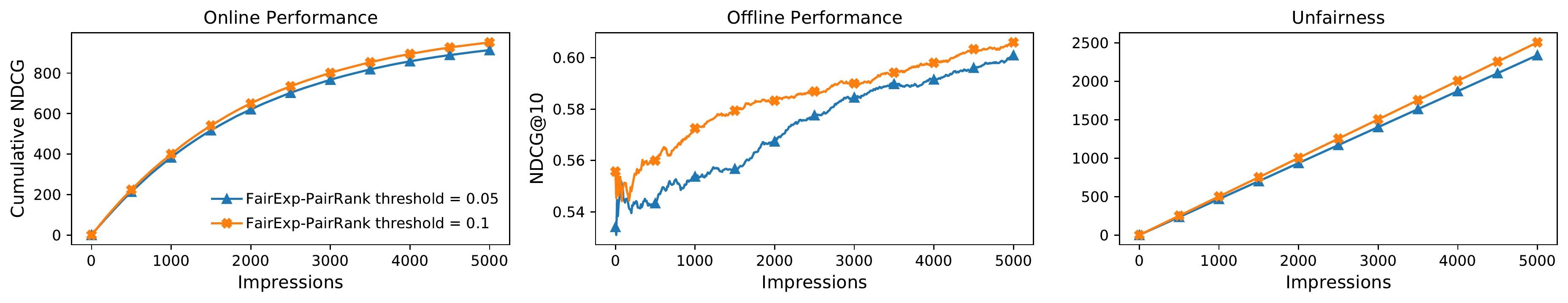}
    \caption{Yahoo dataset with group attribute=feature 470, $\beta=0.477$}
    \label{fig:control471}
  \end{subfigure}
\caption{Effect of fairness control on online and offline performance for Yahoo dataset under informational click model.}
  \label{fig:yahoocontrol}
\end{figure*}

\subsubsection{Trade-off between relevance learning and fairness control.}
In Figure~\ref{fig:web10kcontrol} and Figure~\ref{fig:yahoocontrol}, we demonstrate the trade-off between the relevance learning and fairness during the course of OL2R under different levels of tolerance for unfairness. In these figures, we compared the online ranking, offline ranking, and unfairness performance on FairExp-PairRank model with different thresholds $\epsilon$ on the unfairness. Similar patterns were observed under different click models and thus we only report the results obtained under the informational click models.

We can clearly observe that with a smaller threshold (i.e., stronger constraint on the fairness), FairExp-PairRank had worse online and offline performance compared to a larger threshold. This is expected as with a smaller threshold, more calibrations are needed on top of the original partition generated by PairRank in each round. In other words, the set $\bar\Omega_t$ got smaller with $\epsilon=0.05$ than that under $\epsilon=0.1$; and therefore the finally presented ranking would deviate more from the desired ranking from PairRank. On the other hand, the model will result in larger unfairness with looser control. This is because as long as the fairness constraint is satisfied, \model{} will return the ranked list with the minimum number of violations on the certain rank orders. Besides, as we discussed before, due to the intrinsic incompatibility of the environment, e.g., some queries have very uneven number of candidate documents between groups, the unfairness may be larger than the threshold we set, which is inevitable. For example, in Figure~\ref{fig:web10kcontrol}, for MSLR-Web10K dataset, the cumulative unfairness may reach around 3 in some rounds, though the threshold we set is at most 0.1. And for Yahoo dataset, shown in Figure~\ref{fig:yahoocontrol}, the queries have more imbalanced candidate documents and no ranking can be generated wit respect to the fair templates, where the model will choose a ranking with minimum unfairness. In this circumstance, even if the algorithm completely ignore relevance, it will still be unable to enforce fairness. For example, as shown in Figure \ref{fig:yahoo_inf}, the other algorithms' unfairness, including FairCo, is three times larger than FairExp-PairRank's. In general, FairExp-PairRank will maximize the allowed unfairness to trade for faster relevance estimation and better online result ranking, which further suggests the validity of our FairSwap design to introduce minimum distortion in the original PairRank's online training.

\subsubsection{Zoom into FairExp-PairRank} To further verify the effect of fairness control on the trade-off between exploration and exploitation in PairRank, we zoom into the trace of the number of violations on the certain rank orders during online interactions in Figure~\ref{fig:nv}. We can observe that the fairness control inevitably leads to constant violations on the identified certain rank orders (i.e., added regret) and the impact varies across different group configurations. This also shows another intrinsic incompatibility in fair result ranking: when the optimal utility-driven ranking is not fair, the system has to trade utility for fairness. And this trade-off is beyond the system's control, as it is affected by the distribution of queries that are issued by the users and the associated documents that are produced by the content providers.

\begin{figure}[t]
    \centering
    \includegraphics[width=0.8\linewidth]{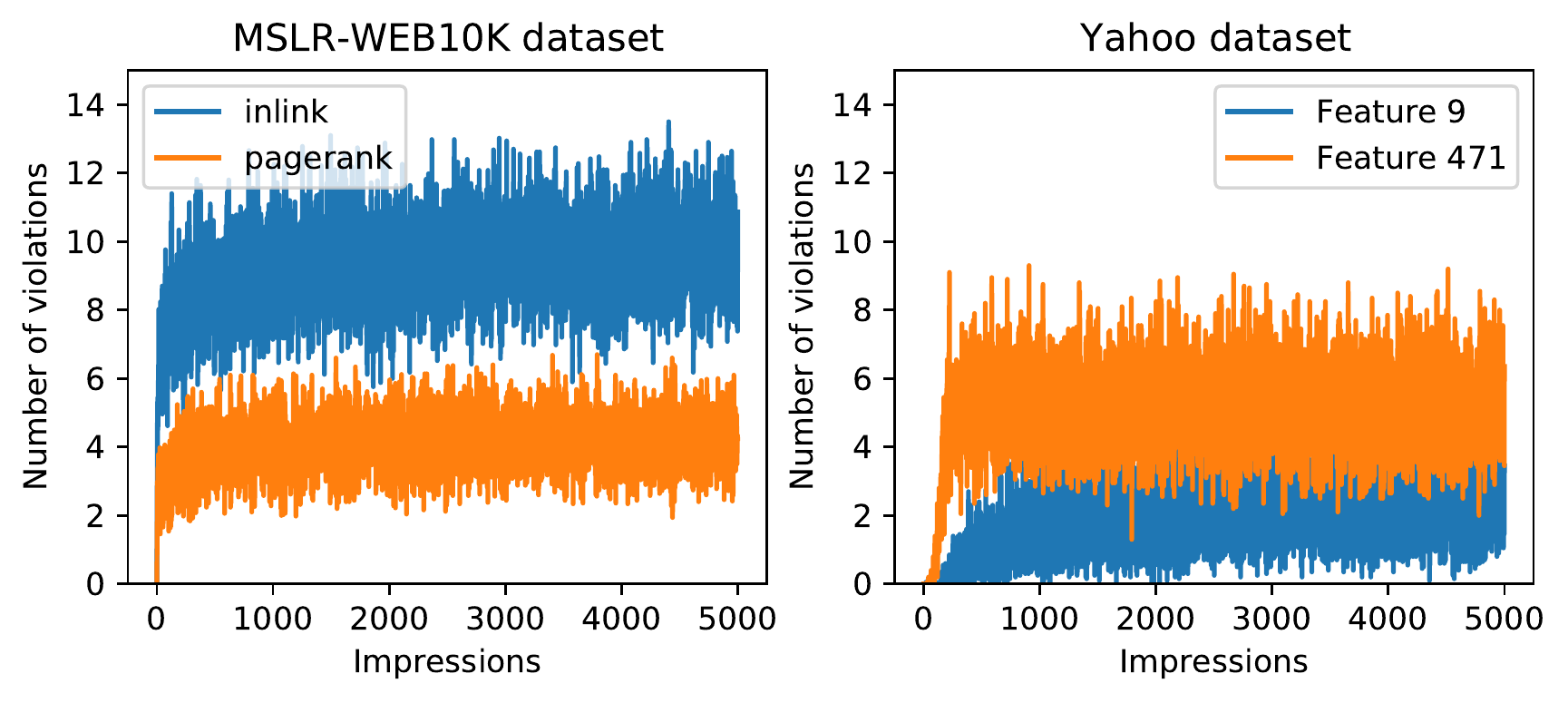}
    \vspace{-5mm}
    \caption{number of violations on certain rank orders.}
    \label{fig:nv}
\end{figure}